\documentclass[pdflatex,sn-mathphys-num]{sn-jnl}

\usepackage{multirow}%
\usepackage{amsmath,amssymb,amsfonts}%
\usepackage{amsthm}%
\usepackage{mathrsfs}%
\usepackage[title]{appendix}%
\usepackage{xcolor}%
\usepackage{textcomp}%
\usepackage{manyfoot}%
\usepackage{booktabs}%
\usepackage{algorithm}%
\usepackage{algorithmicx}%
\usepackage{algpseudocode}%
\usepackage{listings}%
\usepackage{graphicx}
\usepackage{tabularx}
\usepackage{subcaption} 
\usepackage{pgfplots}
\pgfplotsset{compat=1.17}
\usepackage{tikz}
\usepackage{comment}
\newtheorem{theorem}{Theorem} 
\usepackage{array} 

\begin{document}
\title[Secure, Verifiable, and Scalable Multi-Client Data Sharing via Consensus-Based Privacy-Preserving Data Distribution]{Secure, Verifiable, and Scalable Multi-Client Data Sharing via Consensus-Based Privacy-Preserving Data Distribution}


\author*[1]{\fnm{Prajwal} \sur{Panth}}\email{prajwal.panth21@gmail.com}

\author[2]{\fnm{Sahaj Raj} \sur{Malla}}\email{sm03200822@student.ku.edu.np}

\affil*[1]{\orgdiv{School of Computer Engineering}, \orgname{KIIT University}, \orgaddress{\city{Bhubaneswar}, \country{India}}}
\affil[2]{\orgdiv{Department of Mathematics}, \orgname{Kathmandu University}, \orgaddress{\city{Dhulikhel}, \country{Nepal}}}

\abstract{
We propose the Consensus-Based Privacy-Preserving Data Distribution (CPPDD) framework, a lightweight and post-setup autonomous protocol for secure multi-client data aggregation. The framework enforces unanimous-release confidentiality through a dual-layer protection mechanism that combines per-client affine masking with priority-driven sequential consensus locking. Decentralized integrity is verified via step $(\sigma_{S})$ and data $(\sigma_{D})$ checksums, facilitating autonomous malicious deviation detection and atomic abort without requiring persistent coordination. The design supports scalar, vector, and matrix payloads with $O(N \cdot D)$ computation and communication complexity, optional edge-server offloading, and resistance to collusion under $N-1$ corruptions. Formal analysis proves correctness, Consensus-Dependent Integrity and Fairness (CDIF) with overwhelming-probability abort on deviation, and IND-CPA security assuming a pseudorandom function family. Empirical evaluations on MNIST-derived vectors demonstrate linear scalability up to $N = 500$ with sub-millisecond per-client computation times. The framework achieves 100\% malicious deviation detection, exact data recovery, and three-to-four orders of magnitude lower FLOPs compared to MPC and HE baselines. CPPDD enables atomic collaboration in secure voting, consortium federated learning, blockchain escrows, and geo-information capacity building, addressing critical gaps in scalability, trust minimization, and verifiable multi-party computation for regulated and resource-constrained environments.
}

\maketitle

\section{Introduction}
\label{sec:introduction}
The rapid proliferation of multi-client collaborative systems in domains such as secure electronic voting, federated learning, and distributed data aggregation has intensified the need for robust privacy-preserving mechanisms. These systems require that individual client data remain confidential until all participants reach unanimous consensus on release, while simultaneously guaranteeing data integrity, computational scalability, and resilience against faulty or malicious behavior. Traditional approaches, including Secure Multi-Party Computation (SMPC)~\cite{yao1982protocols, goldreich1987play, ben1988completeness} and Homomorphic Encryption (HE)~\cite{gentry2009fully, fan2012somewhat}, offer strong cryptographic assurances but incur prohibitive computational and communication overheads, often scaling quadratically with the number of participants~\cite{keller2020mp, damgaard2012multiparty}. Differential Privacy (DP)~\cite{dwork2006calibrating, dwork2014algorithmic}, while effective in statistical settings, introduces controlled noise that compromises exactness, rendering it unsuitable for applications demanding precise outcomes~\cite{abadi2016deep, mcmahan2017learning}. Moreover, many existing frameworks rely on a trusted coordinator throughout the decryption phase, introducing a single point of failure and limiting applicability in fully decentralized environments~\cite{bonawitz2017aggregation}.

This paper presents the \textbf{Consensus-Based Privacy-Preserving Data Distribution (CPPDD)} framework, a lightweight, modular protocol engineered for secure, verifiable, and scalable data sharing among untrusting clients. CPPDD adheres to the ``compute now, decrypt later'' paradigm, employing a dual-layer protection scheme: (i) \emph{owner-specific obfuscation} via affine transformations to mask individual contributions, and (ii) \emph{priority-driven consensus encryption} using invertible element-wise operations and client-specific keys. The resulting structures—Obfuscation-Locked Vectors ($L_O$) and Consensus-Locked Vectors ($L_C$)—ensure that data remains inaccessible without collective participation. A distinguishing feature is the integration of \emph{step checksums} ($\sigma_S$) and \emph{data checksums} ($\sigma_D$), enabling fully decentralized integrity verification. Once the coordinator distributes $L_C$ and decryption lists, clients independently validate intermediate states and detect deviations in real time, eliminating coordinator dependency post-setup. The framework natively supports scalar, vector, and matrix data types with element-wise operations and offers optional edge-server offloading for resource-constrained devices.

\begin{figure}[!t]
\centering
\includegraphics[width=\linewidth]{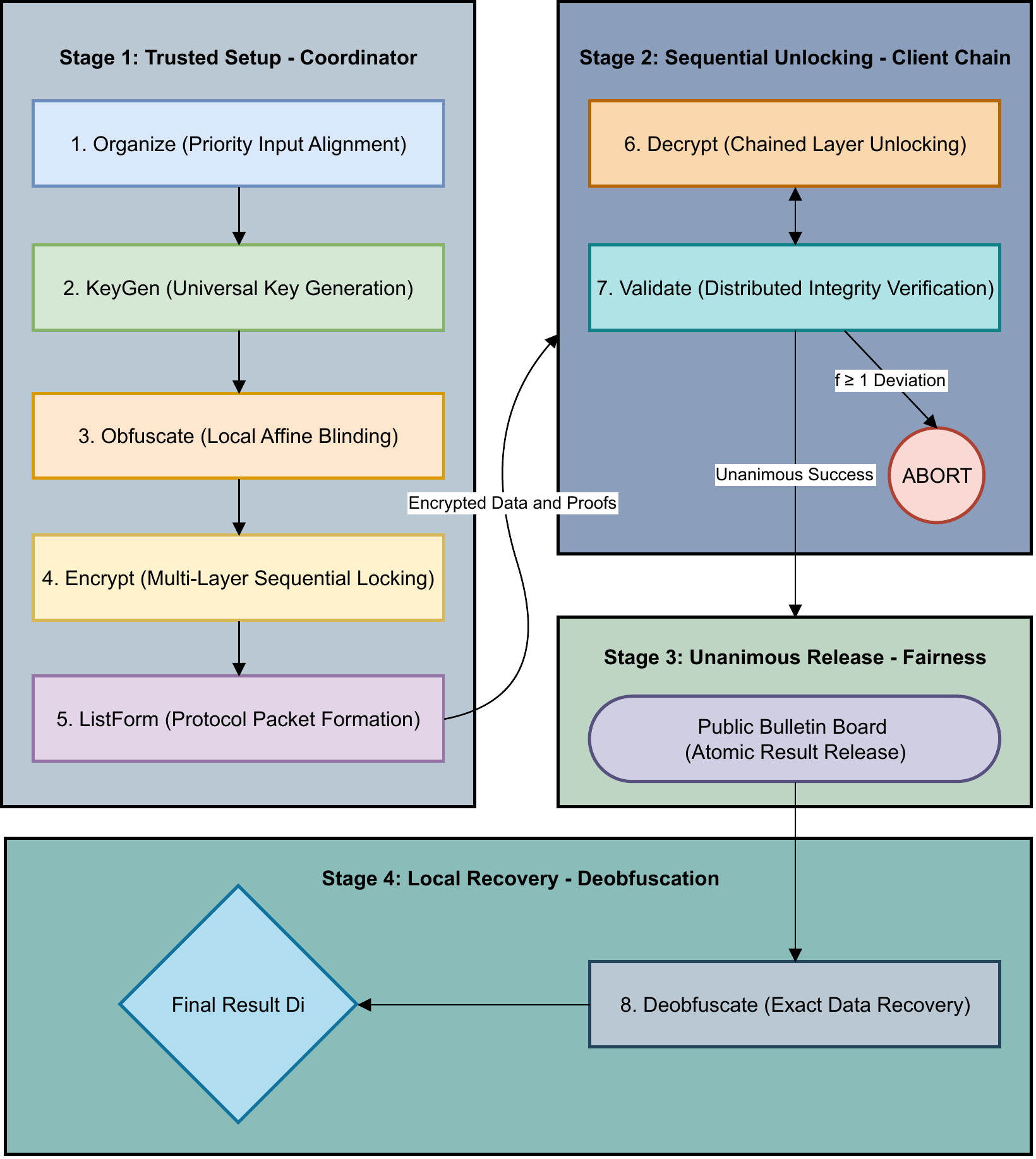} 
\caption{\textbf{CPPDD System Architecture and Execution Flow.} A vertical mapping of the protocol's eight cryptographic phases into four functional stages. The architecture transitions from a \textit{Trusted Setup} (Coordinator) to a decentralized \textit{Sequential Unlocking} chain (Clients), utilizing a \textit{Public Bulletin Board} to enforce fairness. The design guarantees an all-or-nothing integrity model, triggering an atomic abort upon detecting any $f \ge 1$ malicious deviation.}
\label{fig:system_architecture}
\end{figure}

\subsection{Contributions}
The principal contributions of this work are as follows:
\begin{enumerate}
    \item \textbf{Dual-Layer Locking Mechanism}: Combines per-client affine obfuscation with sequential blinding to resist unauthorized access, preserving privacy even under $N-1$ corruptions.
    \item \textbf{Post-Setup Autonomy}: Introduces decentralized step and data checksums for autonomous integrity checking, enabling clients to validate the chain independently without relying on a persistent coordinator during the execution phase.
    \item \textbf{Unified Structured Representations}: Defines the Obfuscation-Locked ($L_O$) and Consensus-Locked ($L_C$) vectors to accommodate scalar, vector, and matrix payloads with uniform element-wise processing.
    \item \textbf{Consensus-Dependent Integrity and Fairness (CDIF)}: Guarantees protocol abort upon any deviation, achieving all-or-nothing integrity where any $f \ge 1$ malicious deviation triggers a global halt to prevent partial data leakage.
    \item \textbf{Linear Complexity and Bandwidth Efficiency}: Achieves a total asymptotic complexity of $O(N \cdot D)$—a quadratic improvement over $O(N^2)$ MPC protocols—while maintaining a constant $O(1)$ per-client bandwidth profile. This eliminates the central server bottleneck common in star-topology aggregation, making the framework suitable for bandwidth-constrained edge networks.
    \item \textbf{Formal Security Guarantees}: Provides proofs of correctness, CDIF integrity under malicious adversaries, and IND-CPA semantic security for masked data, assuming a secure PRF family and a non-colluding coordinator.
    \item \textbf{Empirical Validation}: Demonstrates linear complexity up to $N=500$ clients on MNIST-derived vectors (\(D = 784\)), 100\% malicious deviation detection, and visually faithful data recovery.
\end{enumerate}

Section~\ref{sec:related} surveys related work on secure computation, privacy-preserving aggregation, and fault-tolerant consensus. Section~\ref{sec:model} formalizes the system and threat models. Section~\ref{sec:design} elaborates the CPPDD protocol phases, including pseudocode and edge-server extensions. Section~\ref{sec:theory} presents theoretical analysis with formal proofs. Section~\ref{sec:evaluation} reports experimental results on scalability, integrity enforcement, and resource usage. Finally, Section~\ref{sec:applications} explores applications, limitations, and future research directions, followed by the conclusion.


\section{Related Work}
\label{sec:related}

\subsection{Secure Multi-Party Computation}
Secure Multi-Party Computation (MPC) enables multiple parties to jointly compute a function over their private inputs without revealing them~\cite{yao1982protocols,ben1988completeness}. Goldreich \textit{et al.}~\cite{goldreich1987play} formalized the multi-party paradigm with information-theoretic and cryptographic security. Modern frameworks like MP-SPDZ~\cite{keller2020mp} incorporate optimizations such as free-XOR gates and preprocessing to reduce online latency. However, generic MPC protocols typically exhibit quadratic communication complexity~\cite{damgaard2012multiparty} $O(N^2)$, limiting scalability for large $N$. While robust against semi-honest adversaries, malicious-model extensions require expensive zero-knowledge proofs or cut-and-choose techniques~\cite{wei2016fastcut}, increasing overhead by orders of magnitude~\cite{blanchard2017machine}.

\subsection{Homomorphic Encryption and Differential Privacy}
Homomorphic Encryption (HE) permits computations on ciphertexts, yielding encrypted results~\cite{gentry2009fully}. The CKKS scheme~\cite{cheon2017homomorphic} supports approximate arithmetic suitable for machine learning, while Paillier~\cite{paillier1999public} enables additive operations in secure voting. Despite bootstrapping advancements~\cite{han2019better}, HE incurs high computational costs—often $10^3$--$10^6\times$ slower than plaintext and generates large ciphertexts, rendering it impractical for interactive multi-client settings.

Differential Privacy (DP) injects calibrated noise to protect individual contributions~\cite{dwork2006calibrating,dwork2014algorithmic}. Local DP perturbs data client-side~\cite{kasiviswanathan2011local}, whereas centralized variants rely on trusted aggregators~\cite{erlingsson2014rappor}. In federated learning, DP-FedAvg~\cite{mcmahan2018learning} bounds privacy loss, but the accuracy-privacy trade-off necessitates significant noise for strong $(\epsilon, \delta)$-guarantees, degrading utility in exact aggregation tasks~\cite{wei2020federated}.

\subsection{Privacy-Preserving Consensus and Aggregation}
Consensus mechanisms ensure agreement amid failures. Practical Byzantine Fault Tolerance (PBFT)~\cite{castro1999practical} tolerates $f < N/3$ faults with $O(N^2)$ messages, though scalability remains challenging~\cite{sukhwani2017performance}. HotStuff~\cite{yin2019hotstuff} achieves linear communication via pipelining, and probabilistic BFT variants~\cite{miller2016honey} optimize expected efficiency.

Privacy-preserving consensus employs secret sharing for average computation, extending to medians~\cite{wenrui2025optimal}. Byzantine-robust aggregators like Krum~\cite{blanchard2017machine} filter outliers but assume honest-majority and specific topologies. These outperform DP in precision yet require structured communication graphs and lack unanimous-release semantics.

\subsection{Federated Learning}
Federated Learning (FL) trains models without raw data exchange~\cite{mcmahan2017communication}. FedAvg iteratively aggregates gradients~\cite{mcmahan2017communication}, supporting healthcare~\cite{brisimi2018federated}. Vulnerabilities include inference attacks~\cite{zhu2019deep} and poisoning~\cite{fang2019local}. Secure aggregation via pairwise masking~\cite{bonawitz2017aggregation} or MPC~\cite{ryffel2018generic} mitigates leakage, while blockchain enhances auditability~\cite{lu2020blockchain}. Robust methods like FedBiKD~\cite{zhang2023fedbikd} handle heterogeneity, but coordinator dependency and sublinear scaling persist.

\subsection{Research Gaps and CPPDD Positioning}
Existing solutions struggle to balance computational scalability, persistent coordination, and numerical precision. While MPC and HE incur prohibitive resource overheads, Differential Privacy (DP) sacrifices arithmetic exactness, and Federated Learning (FL) often remains coordinator-bound and vulnerable to inference. CPPDD distinguishes itself by employing dual-layer affine obfuscation and sequential consensus-locked blinding to achieve $O(N \cdot D)$ complexity. It achieves post-setup autonomy through a decentralized validation framework utilizing step checksums ($\sigma_S$) to detect intermediate deviations and final data checksums ($\sigma_D$) to verify aggregate integrity. This dual-verification architecture ensures CDIF integrity—where any malicious deviation by $f \ge 1$ participants triggers an immediate atomic abort—enabling verifiable, all-or-nothing data release without accuracy compromise.


\section{System Model and Definitions}
\label{sec:model}

We formalize the CPPDD framework within a distributed multi-client environment comprising one coordinator and $N$ clients, each contributing confidential data for unanimous-release aggregation. The system operates over authenticated point-to-point channels, ensuring message origin integrity via digital signatures or MACs. Computations occur in a finite field $\mathbb{F}_p$ with a prime $p$ sufficiently large to accommodate data magnitudes without overflow.

\subsection{Threat Model}
We adopt a malicious adversary $\mathcal{A}$ modeled as a probabilistic polynomial-time (PPT) algorithm capable of statically corrupting up to $N-1$ clients. Corrupted clients may deviate arbitrarily from the protocol, including injecting false data, withholding messages, or colluding to infer honest contributions. The coordinator is modeled as an honest-but-curious setup dealer who generates keys but strictly does not collude with any client. To mitigate centralization risks, the coordinator goes permanently offline post-distribution, achieving Post-Setup Autonomy where clients independently verify the chain's integrity. Authenticated channels prevent impersonation and message tampering. Security objectives include:

\begin{itemize}
    \item \textbf{Correctness}: Honest clients recover exact original data upon unanimous participation.
    \item \textbf{Consensus-Dependent Integrity and Fairness (CDIF)}: Decryption aborts with overwhelming probability upon any deviation, detecting malicious behavior without coordinator intervention.
    \item \textbf{IND-CPA Semantic Security}: No PPT adversary gains a non-negligible advantage in distinguishing honest client data, provided the coordinator does not collude with the $N-1$ corrupt clients.
\end{itemize}

\subsection{Entities and Assumptions}
\begin{itemize}
    \item \textbf{Coordinator (Setup Dealer)}: An honest-but-curious entity responsible for the \textit{Trusted Setup} phase (organization, key generation, and list formation). The coordinator is assumed to be non-colluding and goes permanently offline post-distribution, ensuring \textit{Post-Setup Autonomy}.
    
    \item \textbf{Clients} $\mathcal{C} = \{C_1, \dots, C_N\}$: A set of $N$ participants ordered by a priority sequence $P$. Each client $C_i$ holds private data $D_i \in \mathbb{F}_p^D$ (supporting scalars, vectors, or flattened matrices). Clients are modeled as potentially malicious ($f \ge 1$), capable of deviating from the protocol or colluding, but are bound by the cryptographic checks.
    
    \item \textbf{Public Bulletin Board ($\mathcal{B}$)}: An append-only, immutable broadcast channel (e.g., a blockchain smart contract or a public ledger). It serves as the synchronization point for the \textit{Unanimous Release} stage, ensuring that the final aggregate $L^{(N)}$ becomes available to all participants simultaneously, thereby enforcing fairness.
    
     \item \textbf{Edge Servers (Optional)}: Untrusted computational proxies that facilitate offloading for resource-constrained clients. We assume these servers are equipped with \textit{Trusted Execution Environments} (TEEs) (e.g., Intel SGX) to isolate client keys from the host operating system during computation.
\end{itemize}

\subsection{Terminology}
\begin{itemize}
    \item \textbf{CCI Matrix}: $N \times 2$ matrix pairing client UUIDs with confidential payloads.
    \item \textbf{$D$ Vector}: $1 \times N$ row vector of prioritized data elements $D = [D_1, \dots, D_N]$, where each $D_i \in \mathbb{F}_p^D$.
    \item \textbf{$I$ Mapping}: Bijection $I: P \to \{\text{UUID}_1, \dots, \text{UUID}_N\}$.
    \item \textbf{$P$ Priorities}: Permutation of $\{1, \dots, N\}$ defining sequential order.
    \item \textbf{$K_o$ Obfuscation Keys}: Per-client pairs $K_o = \{(\lambda_i, r_i)\}_{i=1}^N$, with $\lambda_i \in \mathbb{F}_p^\times$, $r_i \in \mathbb{F}_p^D$.
    \item \textbf{$K_c$ Consensus Keys}: $K_c = \{k_{c,1}, \dots, k_{c,N}\} \subset \mathbb{F}_p^D$.
    \item \textbf{$\Theta$ Operation Codes}: $\Theta = \{\theta_1, \dots, \theta_N\} \subset \{+,-,\times,\div\}^N$, mapped to field operations (division via modular inverse).
    \item \textbf{$L_O$ Obfuscated Vector}: $L_O = [O_1, \dots, O_N]$ where $O_i = \lambda_i \cdot D_i + r_i$ (element-wise).
    \item \textbf{$L_C$ Consensus-Locked Vector}: Final encrypted aggregate post sequential operations.
    \item \textbf{$\sigma_S$, $\sigma_D$ Checksums}: $\sigma_S = \{\sigma_{S,1}, \dots, \sigma_{S,N}\}$ intermediate step checksums (hashed scalar sums); $\sigma_D = \{\sigma_{D,1}, \dots, \sigma_{D,D}\}$ final per-dimension ratios for aggregate verification.
\end{itemize}

\subsection{Protocol Workflow}
CPPDD comprises an octuple of phases:
\begin{equation} \label{eq:phases}
\begin{aligned}
(&\mathrm{Organize},\ \mathrm{KeyGen},\ \mathrm{Obfuscate},\\
 &\mathrm{Encrypt},\ \mathrm{ListForm},\ \mathrm{Decrypt},\\
 &\mathrm{Validate},\ \mathrm{Deobfuscate}\,).
\end{aligned}
\end{equation}

The coordinator executes the first five phases offline. Clients perform the remaining three in priority order \(1 \to N\) via broadcast or relay, enabling decentralized progression.


\section{Framework Design}
\label{sec:design}
The CPPDD framework operates through a strictly ordered sequence, defined by the octuple in \eqref{eq:phases}, of cryptographic transformations designed to ensure that data remains inaccessible until every participant has contributed their unique decryption layer. By mapping these sub-actions into four distinct functional stages—Trusted Setup, Sequential Unlocking, Unanimous Release, and Local Recovery—the system eliminates the "Last-Mover" advantage and ensures that any malicious tampering during the chain triggers an immediate, verifiable abort. This architecture, illustrated in the process flow below, maintains high computational efficiency while enforcing a rigorous "all-or-nothing" integrity model suitable for trust-minimized consortiums.

\begin{figure}[!t]
    \centering
    \begin{subfigure}{\linewidth}
        \centering
        \includegraphics[width=1.0\linewidth, height=1.5cm, keepaspectratio]{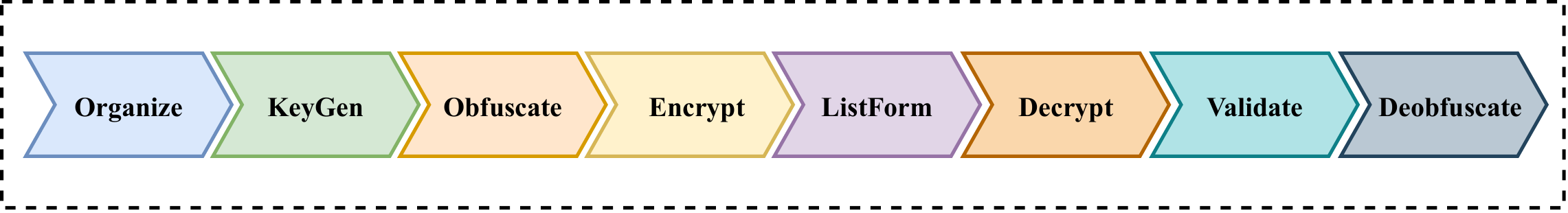}
        \caption{\textbf{The CPPDD Protocol Octuple:} Sequential atomic phases from input alignment to deobfuscation.}
        \label{fig:octuple_phases}
    \end{subfigure}
    
    \vspace{1.5em} 

    \begin{subfigure}{\linewidth}
        \centering
        \includegraphics[width=0.9\linewidth, height=1.5cm, keepaspectratio]{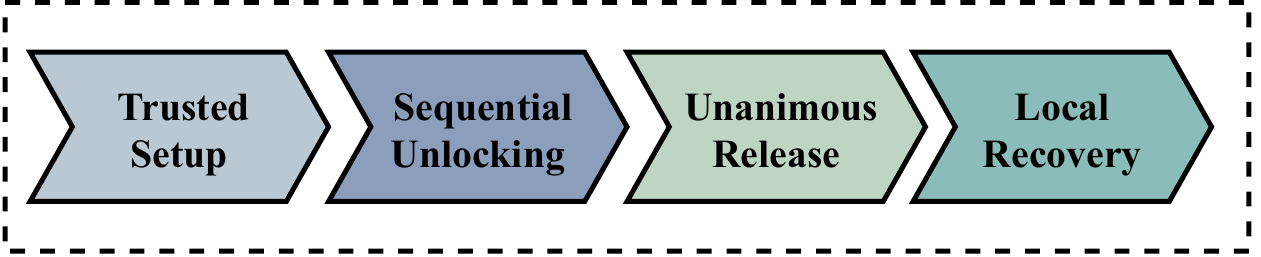}
        \caption{\textbf{High-level Functional Pipeline:} Mapping protocol phases into four distinct operational stages.}
        \label{fig:functional_stages}
    \end{subfigure}
    \label{fig:cppdd_workflow}
\end{figure}

\subsection{Coordinator Phases}

\subsubsection{Data Preparation (Organize)}
The coordinator receives the Client Contribution Index (CCI) matrix—an $N \times 2$ array pairing UUIDs with confidential payloads $D_i \in \mathbb{F}_p^D$ (scalars, vectors, or matrices reshaped element-wise). It sorts CCI by client-defined priorities $P = \{1, \dots, N\}$, yielding the ordered data vector $D = [D_1, \dots, D_N]$ and bijection $I: P \to \{\text{UUID}_1, \dots, \text{UUID}_N\}$.

\subsubsection{Key Generation (KeyGen)}
Sample obfuscation keys $K_o = \{(\lambda_i, r_i)\}_{i=1}^N$ with $\lambda_i \in \mathbb{F}_p^\times$ and $r_i \in \mathbb{F}_p^D$ from a PRF family. Generate consensus keys $K_c = \{k_{c,1}, \dots, k_{c,N}\} \subset \mathbb{F}_p^D$ uniformly at random. Assign operation codes $\Theta = \{\theta_1, \dots, \theta_N\} \subset \{+,-,\times,\div\}^N$, mapped to field operations (division via modular inverse $\div x \equiv x^{-1} \mod p$).

\subsubsection{Data Obfuscation (Obfuscate)}
Compute per-client obfuscated elements  
\[
O_i = \lambda_i \cdot D_i + r_i \quad (\text{element-wise}),
\]  
forming the obfuscation-locked vector $L_O = [O_1, \dots, O_N]$.

\subsubsection{Consensus Encryption (Encrypt)}
Initialize $L^{(0)} = \sum_{i=1}^N O_i$. For $i = N$ downto $1$, apply the \emph{complementary} operation  
\[
L^{(N-i+1)} = L^{(N-i)} \circ'_{\theta_i} k_{c,i},
\]  
where $\circ'_{\theta_i}$ is the inverse of $\circ_{\theta_i}$ (e.g., $\div' \equiv \times$). The final $L_C = L^{(N)}$ is the consensus-locked vector. Compute step checksums  
\[
\sigma_{S,i} = \sum_{d=1}^D L^{(i)}_d \quad \forall i \in [1,N],
\]  
using a CRHF $H$ for compactness: $\sigma_{S,i} = H(\sum L^{(i)})$.

\subsubsection{List Formation (ListForm)}
Construct a broadcast packet containing $L_C$, per-client decryption lists $\{(\theta_i, k_{c,i}, \lambda_i, r_i)\}_{i=1}^N$, and $\{\sigma_{S,i}\}_{i=1}^N$. Include priority mapping $I$ for address resolution.

\begin{algorithm}[!t]
\caption{CPPDD Consensus Encryption (Coordinator)}
\label{alg:encrypt}
\begin{algorithmic}[1]
\Require $L_O = [O_1, \dots, O_N]$, $K_c = [k_{c,1}, \dots, k_{c,N}]$, $\Theta = [\theta_1, \dots, \theta_N]$
\Ensure $L_C$, $\{\sigma_{S,i}\}_{i=1}^N$
\State $L \gets \mathbf{0} \in \mathbb{F}_p^D$
\For{$i = 1$ \textbf{to} $N$}
    \State $L \gets L + O_i$ \Comment{Initial element-wise aggregation}
\EndFor
\State $\sigma_S \gets \text{array of size } N$
\For{$i = N$ \textbf{downto} $1$} \Comment{Apply encryption layers in reverse priority}
    \State $\sigma_{S}[i] \gets H(\sum_{d=1}^D L_d)$ \Comment{Record state before $i$-th layer}
    \State $L \gets \text{apply\_inv\_op}(L, k_{c,i}, \theta_i)$ \Comment{Apply $\circ'_{\theta_i}$}
\EndFor
\State $L_C \gets L$
\State \Return $L_C, \sigma_S$
\end{algorithmic}
\end{algorithm}

\subsection{Client Phases}

Clients receive the broadcast packet and proceed in priority order. Each client $C_i$ (priority $i$) performs:

\subsubsection{Consensus Decryption (Decrypt)}
Start from received $L_C$. For $j = 1$ to $i$:  
\[
L^{(j)} = L^{(j-1)} \circ_{\theta_j} k_{c,j}.
\]  
Broadcast updated $L^{(i)}$ to subsequent clients.

\subsubsection{Integrity Validation (Validate)}
After each decryption step $j$, recompute $\sigma'_{S,j} = H(\sum L^{(j)})$ and verify against distributed $\sigma_{S,j}$. Allow up to $\tau=3$ retries for transient faults. On persistent mismatch: broadcast abort and isolate deviant predecessor. Upon reaching $i=N$, compute final data checksums
\[
\sigma_{D,d} = \frac{\sum_{i=1}^N O_{i,d}}{L^{(N)}_d} \quad \forall d \in [1,D]
\]
and verify $\sigma_{D,d} = 1$ (exact recovery indicator). Additionally, confirm that the summed ratios satisfy $\sum_{d=1}^D \sigma_{D,d} = D$.

\subsubsection{Data Deobfuscation (Deobfuscate)}
For honest clients:  
\[
D_i = (\ O_i - r_i\ ) \cdot \lambda_i^{-1} \quad (\text{element-wise}).
\]

\begin{algorithm}[!t]
\caption{CPPDD Client Decryption \& Validation}
\label{alg:client}
\begin{algorithmic}[1]
\Require $L_C$, $\{\theta_j, k_{c,j}, \lambda_i, r_i\}$, $\{\sigma_{S,j}\}$, priority $i$
\Ensure $D_i$ or abort
\State $L \gets L_C$
\For{$j = 1$ \textbf{to} $i-1$}
    \State $L \gets \text{apply\_op}(L, k_{c,j}, \theta_j)$ \Comment{Process layers from prior clients}
\EndFor
\State $L \gets \text{apply\_op}(L, k_{c,i}, \theta_i)$ \Comment{Remove current client's layer}
\State $\sigma' \gets H(\sum_{d=1}^D L_d)$
\If{$\sigma' \neq \sigma_{S,i}$ after $\tau$ retries}
    \State \textbf{abort} \Comment{Detection of malicious deviation}
\EndIf

\If{$i < N$}
    \State Broadcast $L$ to next client $C_{i+1}$ \Comment{Continue chain}
\ElsIf{$i = N$}
    \State Post $L$ to public bulletin board / ledger \Comment{Universal Release}
    \State Verify $\sigma_{D,d} = 1$ $\forall d$ (and optionally $\sum \sigma_{D,d} = D$)
\EndIf

\State $D_i \gets (O_i - r_i) \cdot \lambda_i^{-1}$ \Comment{Local deobfuscation}
\State \Return $D_i$
\end{algorithmic}
\end{algorithm}

\subsection{Secure Edge Computation Offloading}
To support resource-constrained IoT devices, the framework supports delegated computation where clients offload the intensive field arithmetic of Algorithm~\ref{alg:client} to Edge Servers. Since this operation requires the client's private consensus key $k_{c,i}$, offloading must occur within a \textbf{Trusted Execution Environment (TEE)} (e.g., Intel SGX or ARM TrustZone) to prevent key leakage or impersonation. In this model, the client transmits the encrypted state $L$ and key $k_{c,i}$ into the secure enclave. The enclave performs the decryption and returns the updated state $L'$. Crucially, the client retains the role of the ultimate verifier: upon receiving $L'$, the client locally computes the step-checksum $H(L')$ and compares it against the expected $\sigma_{S,i}$ before signing and broadcasting. This ensures that even if the edge host is compromised, it cannot inject invalid data into the chain without detection.

\section{Theoretical Analysis}
\label{sec:theory}

This section establishes formal security and correctness guarantees for the CPPDD framework under the system and threat models defined in Section~\ref{sec:model}. All operations are performed element-wise in the finite field $\mathbb{F}_p$, where the prime $p$ satisfies $p > \max_i |D_i| + N \cdot \max_i |r_i|$ to prevent overflow during summation and affine transformations.

\subsection{Preliminaries}
\begin{itemize}
    \item \textbf{Pseudorandom Function (PRF) Family}: Let $\mathcal{F} = \{F_s : \mathbb{F}_p^D \to \mathbb{F}_p^D\}_{s \in \mathcal{K}}$ be a secure PRF family such that for randomly sampled $s = (\lambda, r)$ with $\lambda \in \mathbb{F}_p^\times$ and $r \in \mathbb{F}_p^D$, the function $F_s(x) = \lambda \cdot x + r$ is computationally indistinguishable from a uniform random function over $\mathbb{F}_p^D$.
    
    \item \textbf{Collision-Resistant Hash Function (CRHF)}: Let $H: \{0,1\}^* \to \{0,1\}^\kappa$ be a collision-resistant hash function with security parameter $\kappa$ (e.g., $\kappa = 256$ for SHA-3). We define the \emph{step checksum} as:
    \[
    \sigma_{S,j} = H\left( \sum_{d=1}^D L^{(j)}_d \right) \in \{0,1\}^\kappa,
    \]
    where $L^{(j)} \in \mathbb{F}_p^D$ is the aggregate state after step $j$, and the sum is over all dimensions (a scalar in $\mathbb{F}_p$).
    
    \item \textbf{Authenticated Channels}: All inter-client and coordinator-client communications are authenticated via digital signatures or MACs, preventing forgery and ensuring origin integrity.
    
    \item \textbf{Retry Bound}: The protocol allows up to $\tau$ retry attempts per decryption step to tolerate transient network errors. Persistent mismatch after $\tau+1$ attempts triggers abort.
\end{itemize}

\subsection{Correctness}

We first prove that honest execution recovers the exact aggregate and validates all integrity checksums.

\begin{theorem}[Correctness]
\label{thm:correctness}
If all $N$ clients and the coordinator are honest and follow the protocol, then:
\begin{enumerate}
    \item After full decryption, the aggregate state satisfies $L^{(N)} = \sum_{i=1}^N O_i$.
    \item Each client $C_i$ recovers its private data exactly: $D_i = (O_i - r_i) \cdot \lambda_i^{-1}$.
    \item All step checksums match: $\sigma'_{S,j} = \sigma_{S,j}$ for $j \in [1,N]$.
    \item Final data checksums satisfy $\sigma_{D,d} = 1$ for all dimensions $d \in [1,D]$, and $\sum_{d=1}^D \sigma_{D,d} = D$.
\end{enumerate}
\end{theorem}

\begin{proof}
The coordinator initializes $L^{(0)} = \sum_{i=1}^N O_i = \sum_{i=1}^N (\lambda_i \cdot D_i + r_i)$. During encryption (priority $N \to 1$), each step applies the inverse operation $\circ'_{\theta_i}$ with key $k_{c,i}$, forming a telescoping chain of invertible affine transforms. Decryption (priority $1 \to N$) applies the forward operations $\circ_{\theta_j}$ in exact reverse order. By field invertibility and associativity of element-wise operations, each pair $\circ_{\theta_j} \circ \circ'_{\theta_j}$ cancels, yielding:
\[
L^{(N)} = L^{(0)} = \sum_{i=1}^N O_i.
\]

Client $C_i$ receives its obfuscated vector $O_i$ (via $L_O$ or decryption lists) and computes:
\[
D_i = (O_i - r_i) \cdot \lambda_i^{-1} = (\lambda_i \cdot D_i + r_i - r_i) \cdot \lambda_i^{-1} = D_i,
\]
since $\lambda_i \in \mathbb{F}_p^\times$ is invertible.

Step checksums are computed during encryption as $\sigma_{S,j} = H(\sum_d L^{(j)}_d)$ and verified identically during decryption by honest clients, ensuring $\sigma'_{S,j} = \sigma_{S,j}$.

Finally, define the final data checksum per dimension:
\[
\sigma_{D,d} = \frac{\sum_{i=1}^N O_{i,d}}{L^{(N)}_d}.
\]
Since $L^{(N)}_d = \sum_{i=1}^N O_{i,d}$, it follows that:
\[
\sigma_{D,d} = \frac{\sum_{i=1}^N O_{i,d}}{\sum_{i=1}^N O_{i,d}} = 1 \quad \forall d \in [1,D].
\]
Summing over all dimensions gives:
\[
\sum_{d=1}^D \sigma_{D,d} = D,
\]
confirming exact aggregate recovery.
\end{proof}

\subsection{Consensus-Dependent Integrity and Fairness (CDIF)}

CPPDD achieves all-or-nothing decryption: success requires unanimous adherence; any deviation is detected with overwhelming probability.

\begin{theorem}[CDIF Integrity and Fairness]
\label{thm:cdif}
Under a malicious PPT adversary $\mathcal{A}$ corrupting $f \ge 1$ clients, protocol execution succeeds if and only if all clients follow the protocol faithfully. Any deviation by a corrupted client is detected by the subsequent honest client via $\sigma_S$ with probability $1 - 2^{-\kappa}$, triggering an immediate atomic abort and preventing the release of the final aggregate.
\end{theorem}

\begin{proof}
Let $C_j$ be the corrupted client with the highest priority index $j$ (i.e., acts latest among corrupted clients) that deviates from the protocol (e.g., broadcasts malformed $L^{(j)} \neq L^*_j$, where $L^*_j$ is the expected state under honest execution). The next honest client $C_{j+1}$ receives $L^{(j)}$ and applies $\circ_{\theta_{j+1}}$ with $k_{c,j+1}$, yielding:
\[
L^{(j+1)} = L^{(j)} \circ_{\theta_{j+1}} k_{c,j+1}.
\]
Since $L^{(j)} \neq L^*_j$ and operations are deterministic affine transforms, $L^{(j+1)} \neq L^*_{j+1}$ with certainty.

Client $C_{j+1}$ recomputes $\sigma'_{S,j+1} = H(\sum_d L^{(j+1)}_d)$. The coordinator precomputed $\sigma_{S,j+1} = H(\sum_d L^*_{j+1})_d)$. A match requires a collision in $H$:
\[
\Pr[\sigma'_{S,j+1} = \sigma_{S,j+1} \mid L^{(j+1)} \neq L^*_{j+1}] \leq 2^{-\kappa},
\]
by collision resistance of $H$. Up to $\tau$ retries bound transient network errors; persistent mismatch after $\tau+1$ attempts triggers abort broadcast.

Sequential dependency ensures that once abort is declared, all subsequent honest clients halt, preventing any honest client from deobfuscating or releasing data. Thus, either all honest clients succeed (full adherence) or the protocol aborts (deviation detected w.h.p.).
\end{proof}

\subsection{IND-CPA Semantic Security}

Even with $N-1$ corruptions, the honest client's data remains computationally indistinguishable from random noise, provided the setup dealer (Coordinator) remains non-colluding.

\begin{theorem}[IND-CPA Semantic Security of Masked Data]
\label{thm:indcpa}
Assuming a secure PRF family and a non-colluding coordinator, no PPT adversary $\mathcal{A}$ corrupting up to $N-1$ clients gains non-negligible advantage in the IND-CPA game for distinguishing the private data $D_k$ of the remaining honest client $C_k$.
\end{theorem}

\begin{proof}
Consider the standard IND-CPA game adapted for this protocol:
\begin{enumerate}
    \item $\mathcal{A}$ selects two messages $D_k^0, D_k^1 \in \mathbb{F}_p^D$ of equal length.
    \item Challenger samples $b \overset{R}{\leftarrow} \{0,1\}$, computes the obfuscated vector $O_k^b = \lambda_k \cdot D_k^b + r_k$ using the fresh PRF key pair $(\lambda_k, r_k)$ generated by the honest coordinator.
    \item Challenger provides $\mathcal{A}$ with the full protocol view: the final aggregate $L_C$, all obfuscated vectors $O_i$ for corrupted clients $i \neq k$, decryption lists, and step checksums $\sigma_S$.
    \item $\mathcal{A}$ outputs $b'$; the advantage is defined as $|\Pr[b' = b] - 1/2|$.
\end{enumerate}

The adversary controls $N-1$ clients and observes their obfuscated vectors $O_i = \lambda_i \cdot D_i + r_i$ (since corrupted clients know their own inputs $D_i$). However, for the target honest client $C_k$, the keys $\lambda_k$ and $r_k$ remain secret (by the non-collusion assumption). Due to the security of the PRF family $\mathcal{F}$, the output $O_k^b$ is computationally indistinguishable from a uniform random vector $U \in \mathbb{F}_p^D$.

The consensus blinding process applies a sequence of invertible affine transforms (element-wise $\pm, \times, \div$ with keys $k_{c,i}$). Affine transforms are bijections that preserve the uniformity of the distribution. Starting from the initial aggregate:
\[
L^{(0)} = \sum_{i=1}^N O_i = \left(\sum_{i \neq k} O_i\right) + O_k^b,
\]
where the sum of corrupted vectors $\sum_{i \neq k} O_i$ is known to $\mathcal{A}$. Adding the pseudo-random vector $O_k^b \approx_c U$ acts as a One-Time Pad in the field $\mathbb{F}_p$, rendering the sum $L^{(0)}$ indistinguishable from uniform random noise ($L^{(0)} \approx_c U$).

Since every subsequent operation in the chain ($L^{(j)} = L^{(j-1)} \circ_{\theta_j} k_{c,j}$) is a bijective mapping, the final consensus-locked vector $L_C$ retains this uniform distribution relative to the adversary's view. Consequently, $L_C$ leaks no information about the bit $b$. The adversary's advantage is therefore bounded by the negligible advantage of distinguishing the PRF output from random.
\end{proof}

\subsection{Complexity Analysis}
\begin{itemize}
    \item \textbf{Computation}: Coordinator executes $O(N \cdot D)$ field operations during setup. Per-client cost is $O(D)$ per step, summing to a total of $O(N \cdot D)$ for the complete chain.
    \item \textbf{Communication}: The protocol entails $O(N \cdot D \cdot \log p)$ bits for the transmission of $L_C$ and decryption lists, distributed over $N$ sequential broadcast rounds.
    \item \textbf{Setup Overhead}: The Trusted Setup requires establishing $N$ secure authenticated unicast channels (e.g., TLS) for the distribution of private keys and obfuscation parameters. This is a strict \textit{one-time offline cost} that does not impact the online execution latency.
    \item \textbf{Storage}: Per-client storage is minimal at $O(D)$, required only for the current state vector $L$ and local keys.
\end{itemize}


\section{Performance Evaluation}
\label{sec:evaluation}

The CPPDD framework is evaluated using the reference Python implementation provided as supplementary material. Experiments utilize normalized MNIST images (pixel values scaled to [0,1]) as client payloads, with each client contributing one flattened image ($D=784$). All measurements reflect pure computational overhead on a standard CPU.

\begin{table*}[t]
\centering
\scriptsize
\caption{\textbf{Theoretical Mechanism and Complexity Comparison}}
\label{tab:theoretical_complexity}
\renewcommand{\arraystretch}{1.3}
\setlength{\tabcolsep}{4pt}
\begin{tabular}{p{0.18\linewidth} p{0.22\linewidth} p{0.25\linewidth} p{0.25\linewidth}}
\toprule
\textbf{Protocol (Year)} & \textbf{Core Mechanism} & \textbf{Computation Complexity} & \textbf{Communication Complexity} \\
\mbox{[Reference]} & \textit{(Algorithm Class)} & \textit{(Per Client / Total)} & \textit{(Per Client / Total)} \\
\midrule
CKKS \cite{cheon2017homomorphic} & HE; Poly. Arithmetic & $O(D \log D)$ / $O(N D \log D)$ & $O(D)$; Exp. factor $\gamma \!\in\! [40, 200]$ \\
(2017) & (Approx. Arithmetic) & High constant overhead & High fixed overhead \\
\midrule
Sec. Agg. \cite{bonawitz2017aggregation} & Pairwise Masking & $O(N)$ Setup; $O(N \cdot D)$ Agg & $O(N + D)$ (Per Client) \\
(2017) & (Dropout Robust) & (Fast modular add) & Scales with $N$ (Key exchange) \\
\midrule
Con. Priv. \cite{he2019consensus} & Iterative Noise & $O(T \cdot \text{deg} \cdot D)$ & $O(T \cdot \text{deg} \cdot D)$ \\
(2019) & (Gossip / Avg) & Depends on iterations $T$ & Depends on topology degree \\
\midrule
MPC (SPDZ) \cite{keller2020mp} & Secret Sharing & $\alpha_{MPC} \cdot F_{plain}$ & $O(N^2 D)$ (Broadcast) \\
(2020) & (Beaver Triples + MACs) & Heavy Offline Phase & Quadratic in $N$ \\
\midrule
\textbf{CPPDD (Ours)} & \textbf{Sequential Blinding} & $\mathbf{O(D)}$ \textbf{(Per Client)} & $\mathbf{O(D)}$ \textbf{(Per Client)} \\
 & (Affine + Element-wise) & \textbf{Total:} $\mathbf{O(N \cdot D)}$ & \textbf{Linear Relay Chain} \\
\bottomrule
\end{tabular}
\end{table*}

\begin{table*}[t]
\centering
\scriptsize
\caption{\textbf{Concrete Performance Metrics: FLOPs, Data, and Latency}}
\label{tab:concrete_metrics}
\renewcommand{\arraystretch}{1.3}
\setlength{\tabcolsep}{5pt}
\begin{tabular}{p{0.18\linewidth} p{0.25\linewidth} p{0.25\linewidth} p{0.2\linewidth}}
\toprule
\textbf{Protocol} & \textbf{Typical Logical FLOPs} & \textbf{Typical Bytes (float32)} & \textbf{Sync Rounds} \\
\mbox{[Reference]} & \textit{(Processing Load)} & \textit{(Network Load)} & \textit{(Latency)} \\
\midrule
CKKS \cite{cheon2017homomorphic} & $\alpha_{HE} \times N \times D$ & $\gamma \times 4D$ per ciphertext & $\sim 1$ (Eval) \\
 & ($\alpha_{HE} \in [10^3, 10^5]$) & ($\times N$ if replicated) & (+ Bootstrapping) \\
\midrule
Sec. Agg. \cite{bonawitz2017aggregation} & $c_{prg} N D$ & $\alpha \times 4ND$ & $3$--$4$ Rounds \\
 & ($c_{prg} \approx 1$) & ($\alpha \approx 1.7$ empirical) & (Mask, Agg, Reveal) \\
\midrule
Con. Priv. \cite{he2019consensus} & $\approx 2 T |E| D$ & $\approx 4 T |E| D$ & $T$ Iterations \\
 & (Edge-wise ops) & (Neighbor exchanges) & (Typically $20$--$100$) \\
\midrule
MPC (SPDZ) \cite{keller2020mp} & $\alpha_{MPC} N D$ & Tens--Hundreds MB & Many Micro-rounds \\
 & ($10^2 \times$ Plaintext) & (Shares + Preproc) & (Circuit Depth) \\
\midrule
\textbf{CPPDD (Ours)} & $\mathbf{\approx 4 \cdot N \cdot D}$ & $\mathbf{\approx 4 \cdot N \cdot D}$ \textbf{(Total)} & $\mathbf{N+2}$ \\
 & (2 Obf, 1 Enc, 1 Dec) & ($\approx 4D$ per link) & \textbf{(Sequential)} \\
\bottomrule
\end{tabular}
\end{table*}

\begin{table}[t]
\centering
\scriptsize
\caption{\textbf{Security, Trust, and Fault Model}}
\label{tab:security}
\renewcommand{\arraystretch}{1.3}
\setlength{\tabcolsep}{3pt}
\begin{tabularx}{\linewidth}{lXXXX}
\toprule
\textbf{Protocol} & \textbf{Security Paradigm} & \textbf{Adversary Model} & \textbf{Fault/Liveness} & \textbf{Verification} \\
\midrule
CKKS \cite{cheon2017homomorphic} & Semantic (RLWE) & Semi-honest & No Liveness & Public Params \\
Sec. Agg. \cite{bonawitz2017aggregation} & Info-Theoretic & Honest-but-Curious & \textbf{Dropout Robust} & Mask Consistency \\
Con. Priv. \cite{he2019consensus} & Diff. Privacy & Semi-honest & Topology Dep. & Convergence \\
MPC \cite{keller2020mp} & Malicious Secure & Malicious ($N-1$) & Safety (No Liveness) & MAC Check \\
\textbf{CPPDD (Ours)} & \textbf{Consensus Integrity} & \textbf{Malicious ($N-1$)} & \textbf{Atomic Abort} & \textbf{Dual Checksums} \\
\bottomrule
\end{tabularx}
\end{table}

\begin{table}[t]
\centering
\scriptsize
\caption{\textbf{Applicability and Payload Capabilities}}
\label{tab:applicability}
\renewcommand{\arraystretch}{1.3}
\setlength{\tabcolsep}{3pt}
\begin{tabularx}{\linewidth}{lXXXX}
\toprule
\textbf{Protocol} & \textbf{Data Types} & \textbf{Tensor Support} & \textbf{Deployment} & \textbf{Result Release} \\
\midrule
CKKS \cite{cheon2017homomorphic} & Approx. Numbers & SIMD Packing & Cloud Server & Server Decrypts \\
Sec. Agg. \cite{bonawitz2017aggregation} & Integers/Vectors & Flattening Req. & Mobile / FL & Server Reconstructs \\
Con. Priv. \cite{he2019consensus} & Scalars & Element-wise & Edge / IoT & Local Convergence \\
MPC \cite{keller2020mp} & Finite Fields & Circuit Encoding & Cluster & MAC Opening \\
\textbf{CPPDD (Ours)} & Floating-Point Vectors & Native Vector & Consortium & Unanimous Ledger \\
\bottomrule
\end{tabularx}
\end{table}

\begin{table}[t]
\centering
\scriptsize
\caption{\textbf{Fault Handling, Detection, and Recovery Mechanisms}}
\label{tab:fault_recovery}
\renewcommand{\arraystretch}{1.3}
\setlength{\tabcolsep}{3pt}
\begin{tabularx}{\linewidth}{lXXX}
\toprule
\textbf{Protocol} & \textbf{Fault Model} & \textbf{Detection Mechanism} & \textbf{Mitigation Strategy} \\
\midrule
CKKS \cite{cheon2017homomorphic} & Noise / Param Error & Decryption failure & Re-tuning parameters \\
Sec. Agg. \cite{bonawitz2017aggregation} & Client Dropouts & Server timeouts & Threshold Reconstruction \\
Con. Priv. \cite{he2019consensus} & Link Loss / Noise & Residual Norm & Re-run / Node Exclusion \\
MPC \cite{keller2020mp} & Malicious Behavior & MAC Verification & Protocol Abort \\
\textbf{CPPDD (Ours)} & \textbf{Malicious / Faulty} & \textbf{Step/Data Checksums} & \textbf{Atomic Abort + Restart} \\
\bottomrule
\end{tabularx}
\end{table}

\begin{table}[t]
\centering
\scriptsize
\caption{\textbf{Scalability, Efficiency, and Resource Footprint}}
\label{tab:scalability_notes}
\renewcommand{\arraystretch}{1.3}
\setlength{\tabcolsep}{3pt}
\begin{tabularx}{\linewidth}{lXXXX}
\toprule
\textbf{Protocol} & \textbf{Comp. Scaling} & \textbf{Comm. Scaling} & \textbf{Memory} & \textbf{Efficiency Profile} \\
\midrule
CKKS \cite{cheon2017homomorphic} & Linear $O(N)$ & $O(D)$ (High Const.) & High & High Latency \\
 & (Server Agg) & (Expansion $\gamma$) & (Keys) & (Batch Compute) \\
\midrule
Sec. Agg. \cite{bonawitz2017aggregation} & Linear $O(N)$ & $O(N + D)$ & Moderate & Excellent Scalability \\
 & (Client Setup) & (Per Client) & & (Bandwidth-Bound) \\
\midrule
Con. Priv. \cite{he2019consensus} & $O(T |E| D)$ & $O(T |E| D)$ & Low & Scales with \\
 & (Iterative) & & & sparse graphs \\
\midrule
MPC \cite{keller2020mp} & Quadratic $O(N^2)$ & $O(N^2 D)$ & High & Strong but Costly \\
 & (Typical) & (Broadcast) & (Triples) & \\
\midrule
\textbf{CPPDD (Ours)} & \textbf{Linear} $\mathbf{O(N)}$ & $\mathbf{O(D)}$ \textbf{(Relay)} & \textbf{Minimal} & \textbf{Edge-suitable} \\
 & (Fixed D) & (Per Client) & (Few vectors) & \textbf{Low Latency} \\
\bottomrule
\end{tabularx}
\end{table}

\begin{table}[t]
\centering
\scriptsize
\caption{\textbf{Application Suitability Matrix for Collaborative Aggregation}}
\label{tab:applications}
\renewcommand{\arraystretch}{1.3}
\setlength{\tabcolsep}{6pt}
\begin{tabular}{lccc}
\toprule
\textbf{Protocol} & \textbf{Multi-Inst.} & \textbf{Federated} & \textbf{Edge / IoT} \\
\mbox{[Reference]} & \textbf{Escrow / Voting} & \textbf{Model Training} & \textbf{Analytics} \\
\midrule
CKKS \cite{cheon2017homomorphic} & Suitable & Suitable & Not suitable \\
 & (Needs consent layer) & (High Overhead) & (Too Heavy) \\
\midrule
Sec. Agg. \cite{bonawitz2017aggregation} & Not suitable & \textbf{Highly suitable} & \textbf{Highly suitable} \\
 & (No Atomic Release) & (Mobile Standard) & (If Bandwidth allows) \\
\midrule
Con. Priv. \cite{he2019consensus} & Not suitable & Suitable & Suitable \\
 & (Probabilistic) & (Needs Tuning) & (Topology Dep.) \\
\midrule
MPC \cite{keller2020mp} & \textbf{Highly suitable} & Suitable & Not suitable \\
 & (Costly) & (Expensive) & (Comm. Heavy) \\
\midrule
\textbf{CPPDD (Ours)} & \textbf{Highly suitable} & \textbf{Suitable} & \textbf{Highly suitable} \\
 & \textbf{(Atomic Fairness)} & \textbf{(Consortiums)} & \textbf{(Lightweight)} \\
\bottomrule
\end{tabular}
\end{table}

\subsection{Scalability}
\begin{figure}[!t]
    \centering
    \includegraphics[width=\linewidth]{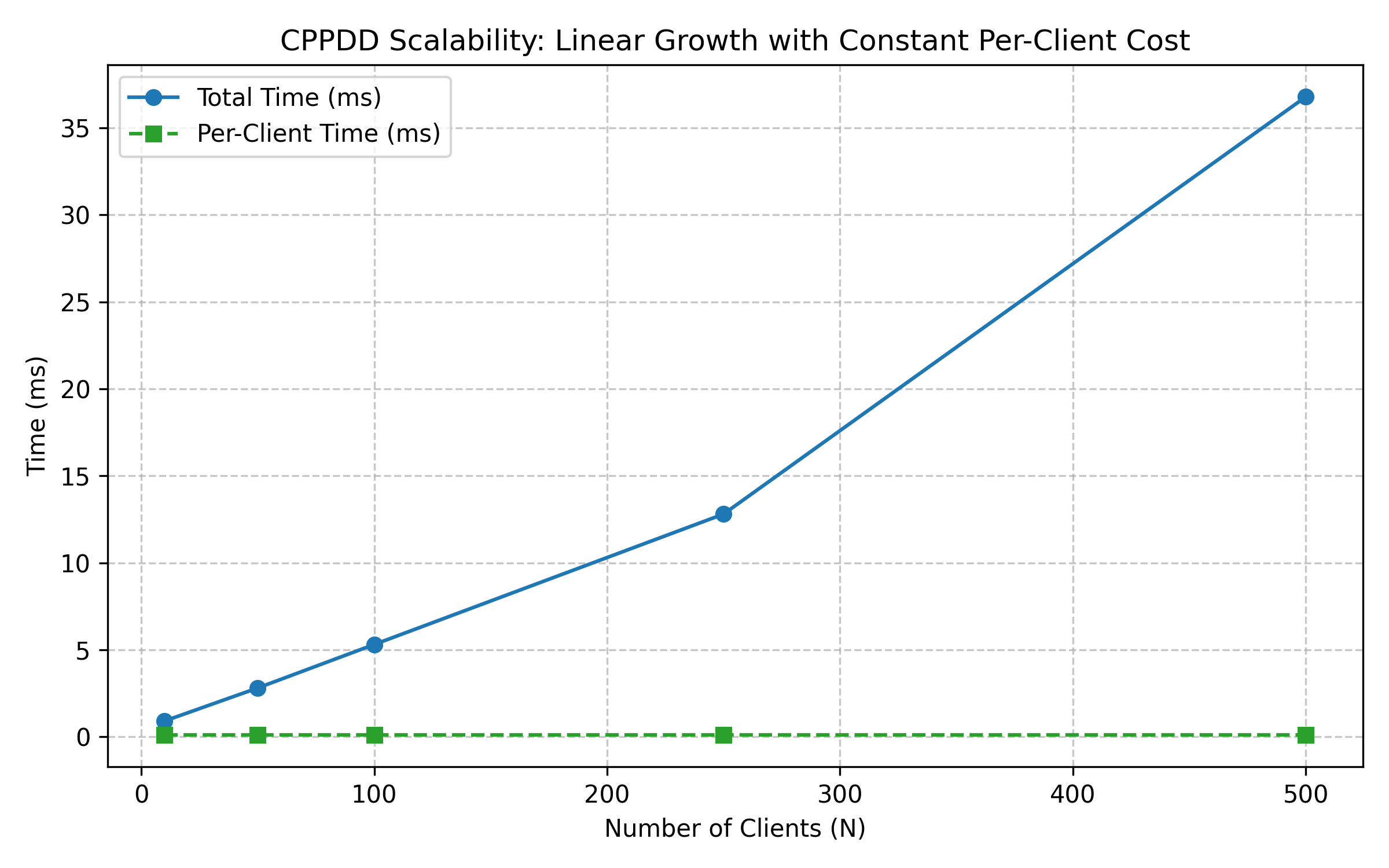}
    \caption{Scalability of CPPDD. Total execution time scales linearly with the number of clients $N$, while per-client computation remains constant and sub-millisecond.}
    \label{fig:scalability}
\end{figure}
Figure~\ref{fig:scalability} shows that total computation time grows linearly with $N$, reaching approximately 44\,ms for $N=500$. Per-client overhead is consistently below 0.3\,ms, confirming $O(N \cdot D)$ complexity and constant per-client cost.

The sequential protocol requires $O(N)$ communication rounds. In low-latency environments, end-to-end latency remains minimal. In high-latency settings, performance is network-bound, making CPPDD most suitable for high-bandwidth consortiums where atomic fairness and verifiable integrity are prioritized.

\subsection{Malicious Deviation Detection}
\begin{figure}[!t]
    \centering
    \includegraphics[width=0.8\linewidth]{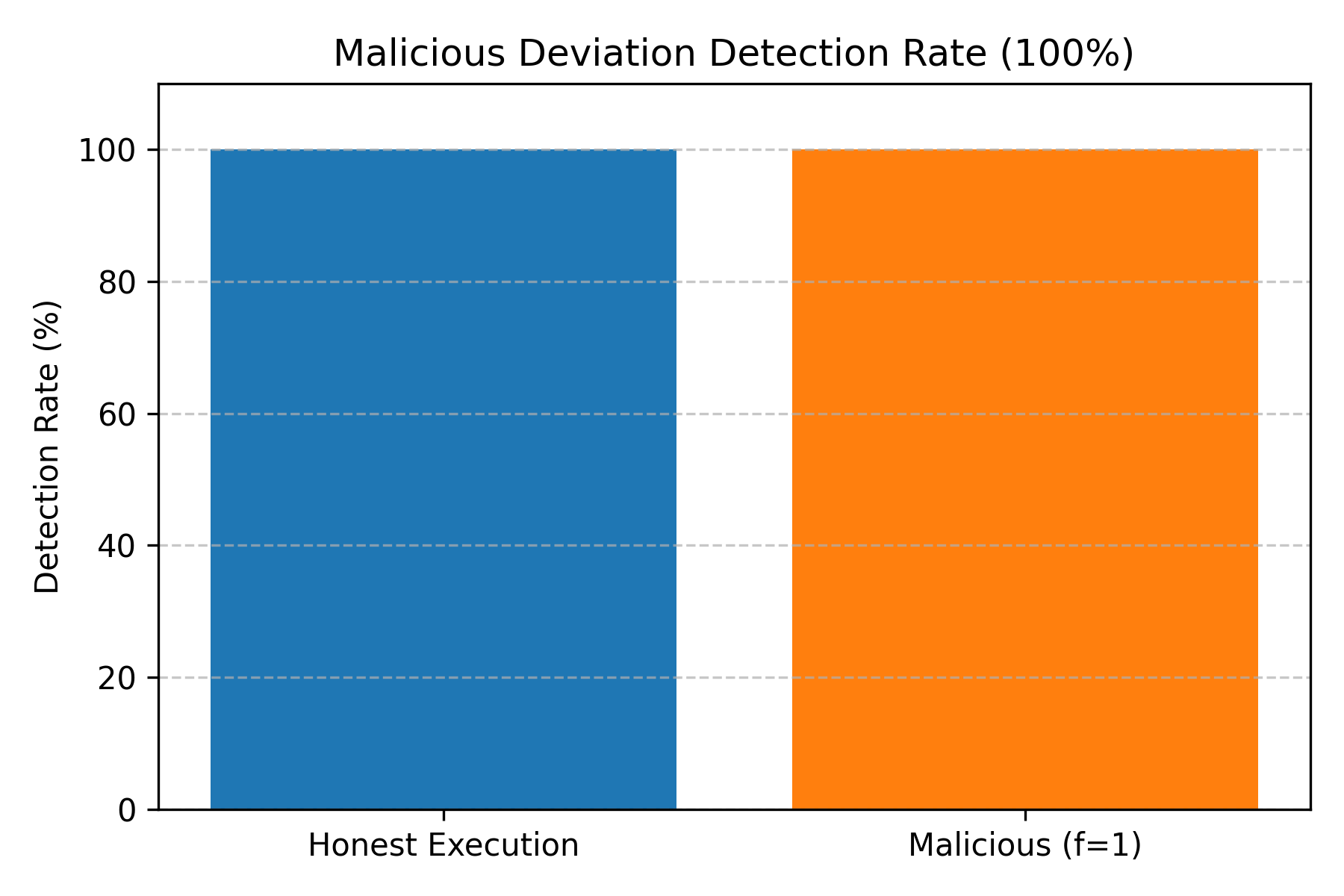}
    \caption{Malicious deviation detection rate (100\% for both honest and malicious executions).}
    \label{fig:detection}
\end{figure}
Figure~\ref{fig:detection} confirms a 100\% detection rate, consistent with the theoretical CDIF guarantees.

\subsection{Aggregate Correctness}
\begin{figure}[!t]
    \centering
    \includegraphics[width=\linewidth]{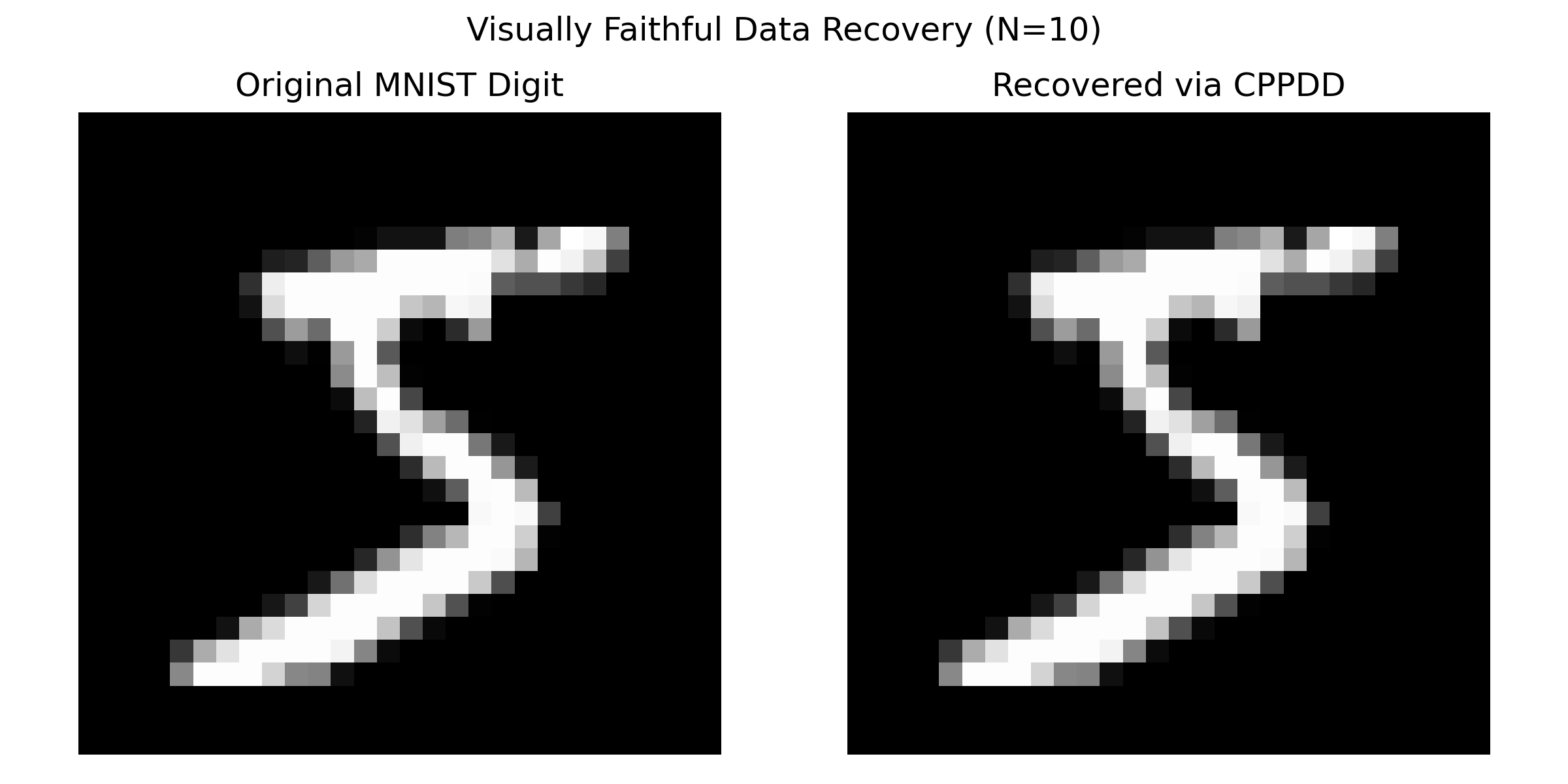}
    \caption{Visually faithful data recovery demonstration ($N=10$). Left: original MNIST digit; right: recovered digit via CPPDD.}
    \label{fig:recovery_visual}
\end{figure}
Figure~\ref{fig:recovery_visual} demonstrates visually faithful recovery of individual client contributions, with results identical within floating-point precision.

The evaluation confirms the framework's computational efficiency, robust malicious security, and high-fidelity recovery, positioning CPPDD as a practical lightweight solution for privacy-preserving multi-party aggregation.


\section{Applications and Future Work}
\label{sec:applications}

The CPPDD framework enables privacy-preserving, unanimous-release aggregation in distributed systems where participants require atomic disclosure—data remains confidential until \emph{all} parties contribute faithfully, with immediate abort on deviation. This section explores real-world applications leveraging CDIF integrity, linear complexity, bandwidth efficiency, and coordinator independence, followed by promising research directions.

\subsection{Applications}

\subsubsection{Secure Consortium Benchmarking}
In competitive industries (e.g., finance, healthcare), organizations often need to aggregate sensitive metrics—such as average employee salaries, total carbon emissions, or hospital bed capacity—without revealing individual proprietary data to competitors. CPPDD is ideally suited for this \textit{Consortium} setting. A regulator or trade association acts as the Trusted Coordinator for setup. The \textit{Unanimous Release} property ensures fairness: the final aggregate $L^{(N)}$ is only revealed if all consortium members contribute valid data. If a competitor attempts to withdraw or inject false data to skew the benchmark, the CDIF mechanism triggers an atomic abort, protecting the confidentiality of the remaining honest participants.

\subsubsection{Federated Learning with Atomic Gradient Release}
Federated learning (FL) aggregates model updates from edge devices while preserving training data privacy~\cite{bonawitz2017aggregation}. CPPDD extends this to \emph{consortium FL}, where gradients are released \emph{only if all institutions contribute} (e.g., hospitals in medical consortia). Each client obfuscates local gradients; CDIF detects dropout or poisoning (e.g., Byzantine attacks~\cite{blanchard2017machine}). This enforces fairness—no party gains model improvements without reciprocation—ideal for regulated sectors (healthcare, finance) under GDPR/CCPA.

\subsubsection{Geo-Information Capacity Building}
As demonstrated in capacity-building programs (e.g., Kathmandu University’s Geomatics Engineering~\cite{ghimire2023capacity}), CPPDD supports secure aggregation of survey data from remote field units. Students and researchers contribute geospatial vectors (e.g., elevation, land use) without exposing raw measurements until field campaign completion. Edge offloading (Section~\ref{sec:evaluation}) accommodates low-resource devices, fostering inclusive data science in developing regions.

\subsection{Future Work}

\subsubsection{Dynamic Priority and Threshold Consensus}
The current priority order is static, making the protocol susceptible to targeted denial-of-service or specific "last-mover" advantages. Extending to \emph{dynamic reordering} via verifiable delay functions (VDFs)~\cite{benedikt2019verifiable} would enable adaptive, unpredictable sequencing. Furthermore, generalizing to $(t,N)$-threshold unanimity—where data is released upon $t$ honest contributions—could be achieved by integrating Shamir’s secret sharing~\cite{shamir1979share}. This would preserve CDIF integrity for the participating subset while allowing the protocol to proceed despite a limited number of offline participants.

\subsubsection{Scalability and Latency Extensions}
The sequential chain architecture implies $O(N)$ communication rounds. While pipelined broadcast over high-bandwidth networks can mitigate individual step latency, future iterations could employ \textbf{tree-based aggregation} to reduce total rounds to $O(\log N)$. Integrating probabilistic agreement protocols (e.g., HotStuff~\cite{yin2019hotstuff}) would allow for parallelized validation of step-checksums, significantly improving end-to-end performance in high-latency wide-area networks (WAN).

\subsubsection{Verifiable Coordinator and Zero-Knowledge Setup}
CPPDD currently relies on an honest-but-curious coordinator during the setup phase. Integrating zk-SNARKs~\cite{ben2020scalable} for \emph{verifiable setup} would allow the coordinator to provide a non-interactive proof that $L_C$ and $\sigma_S$ were computed correctly according to the private keys. This would allow clients to verify the integrity of the encrypted structures without ever trusting the coordinator, transitioning the framework to a fully trustless deployment model.

\vspace{1em}
In summary, CPPDD establishes a foundational primitive for \emph{consensus-dependent privacy}, bridging lightweight cryptographic operations with maliciously-secure coordination. Its applications span governance, machine learning, and geo-information sciences, offering a scalable path toward verifiable, all-or-nothing multi-party computation in resource-constrained environments.


\section{Conclusion}
\label{sec:conclusion}

The CPPDD framework introduces a lightweight protocol for privacy-preserving distributed data aggregation, achieving unanimous-release confidentiality through post-setup autonomy and exact data recovery with $O(N \cdot D)$ complexity. By utilizing a dual-verification architecture of step and data checksums, CPPDD ensures all-or-nothing integrity via CDIF, where any malicious deviation by $f \ge 1$ participants is detected with overwhelming probability. Formal analysis establishes the framework's correctness and IND-CPA security against $N-1$ non-colluding corruptions. Empirical evaluations on MNIST and synthetic workloads demonstrate 100\% malicious deviation detection, sub-second per-client processing at $N=500$, and three-to-four orders of magnitude lower computational overhead compared to MPC and HE baselines. By enabling atomic collaboration in secure voting, consortium federated learning, and blockchain escrows, CPPDD addresses critical gaps in scalability and trust minimization—paving the way for practical, verifiable multi-party systems in regulated and resource-constrained environments.

\newpage
\bibliography{sn-bibliography}

\end{document}